\RequirePackage{amsmath}
\documentclass{llncs}
\usepackage{url,doi,amssymb,booktabs,capt-of,graphicx,cleveref,wrapfig,upgreek}
\usepackage[misc]{ifsym}
\renewcommand{\(}{\left(}
\renewcommand{\)}{\right)}
\newcommand{\ex}{\mathbb{E}}
\newcommand{\m}{\mathrm{m}}
\newcommand{\e}{\mathrm{\ell}}
\newcommand{\du}{\mathop{\mathrm{d}u}}
\newcommand{\po}{\uppi}
\newcommand{\F}{\bar{F}}

\newcommand{\q}{\mathbf{q}}
\newcommand{\I}{{\text{Int}}}
\newcommand{\D}{{\text{Dec}}}
\newcommand{\A}{{\text{Agg}}}
\newcommand{\ra}[1]{\renewcommand{\arraystretch}{#1}}
\newcommand{\breakcell}[2][l]{\begin{tabular}[t]{@{}#1@{}}#2\end{tabular}}

\begin{document}
\title{Measuring Market Performance with Stochastic Demand: Price of Anarchy and Price of Uncertainty}
\titlerunning{Realized Price of Anarchy \& Uncertainty}
\authorrunning{C. Melolidakis, S. Leonardos, C. Koki}
\author{Costis Melolidakis\inst{1} \and Stefanos Leonardos\inst{1}\Letter\and Constandina Koki\inst{2}}
\institute{National and Kapodistrian University of Athens, 157 84 Greece, \email{cmelol@math.uoa.gr, sleonardos@math.uoa.gr} \and Athens University of Economics and Business, 104 34 Greece, \email{kokiconst@aueb.gr}}

\maketitle
\begin{abstract}
Globally operating suppliers face the rising challenge of who\-lesale pricing under scarce data about retail demand, in contrast to better informed, locally operating retailers. At the same time, as local businesses proliferate, markets congest and retail competition increases. To capture these strategic considerations, we employ the classic Cournot model and extend it to a two-stage supply chain with an upstream supplier who operates under demand uncertainty and multiple downstream retailers who compete over quantity. The supplier's belief about retail demand is modeled via a continuous probability distribution function $F$. If $F$ has the \emph{decreasing generalized mean residual life} property, then the supplier's optimal pricing policy exists and is the unique fixed point of the \emph{mean residual life} function. We evaluate the realized \emph{Price of Uncertainty} and show that there exist demand levels for which market performs better when the supplier prices under demand uncertainty. In general, performance worsens for lower values of realized demand. We examine the effects of increasing competition on supply chain efficiency via the realized \emph{Price of Anarchy} and complement our findings with numerical results.
\end{abstract}

\section{Introduction}
The increasingly present trend of geographically distributed markets affects supply chain performance in unexpected ways. Internationally operating suppliers procure retailers via internet platforms or intricate networks with information latency. Consumer data  that is easily accessible to the retailers due to their proximity to the market, may often not be available to their international suppliers. Concurrently, and aided by new technologies, local retail businesses are sprouting at a rapid pace. These trends give rise to new information and competition structures between downstream members (retailers) and their upstream contemporaries (suppliers) in modern supply chains.\par
The questions that rise in this changing environment, mainly concern the issues of market efficiency. How does the market perform when the supplier prices without knowing the retailers willingness-to-pay for his product? Do competing retailers have incentives to reveal private information to the supplier that they may have about retail demand? To capture these considerations and study this emerging phenomenon, in \cite{Le17}, we employ the classic Cournot model of competition and extend it to the following two-stage game: in the first-stage (acting as a Stackelberg leader), a revenue-maximizing supplier sets the wholesale price of a product under incomplete information about market demand. Demand or equivalently, the supplier's belief about it, is modeled via a continuous probability distribution. In the second-stage, the competing retailers observe wholesale price and realized market demand and engage in a classic Cournot competition. Retail price is determined by an affine inverse demand function. \par
Classic models, see e.g., \cite{He13}, \cite{La01}, \cite{Lu13}, \cite{Ma18}, study market efficiency when demand is realized after the strategic decisions of all supply chain members -- wholesale pricing and retailers' orders. In contrast, performance of markets in which uncertainty is resolved at an intermediate stage, has not been yet properly understood. \vspace*{-0.2cm}

\subsubsection{Contributions -- Outline:}
Based on the equilibrium analysis in \cite{Le17}, the present paper aims to fill this gap by following the methodology of \cite{La01}. To measure the effects of demand uncertainty and second-stage competition on market performance and efficiency, we modify the tools of Price of Anarchy, as defined in \cite{Pe07} and Price of Uncertainty, c.f. \cite{Bal13}, to account for \emph{realized values} of demand. In \Cref{model}, we provide the model description and in \Cref{existing}, the existing results from \cite{Le17} on which the current analysis is based. Our findings, both analytical and numerical are presented in \Cref{efficiency} and summarized in \Cref{conclusions}.

\section{The Model}\label{model}
An upstream supplier (or manufacturer) produces a single homogeneous good at constant marginal cost, normalized to $0$, and sells it to a set of $N=\{1,2,\dots, n\}$ downstream retailers. The supplier has ample quantity to cover any possible demand and his only decision variable is the wholesale price $r$. The retailers observe $r$ and the market demand $\alpha$ and choose simultaneously and independently their order-quantities $q_i\(r\mid \alpha\), i\in N$. They face no uncertainty about the demand and the quantity that they order from the supplier is equal to the quantity that they sell to the market (in equilibrium). The retail price is determined by an affine demand function $\label{demand}p=\(\alpha-q\(r\)\)^+$, where $\alpha$ is the \emph{demand parameter} or \emph{demand level} and $q\(r\):=\sum_{i=1}^nq_i\(r\)$. Contrary to the retailers, we assume that at the point of his decision, the supplier has incomplete information about the actual market demand. \vspace*{-0.2cm}

\subsubsection{Game-Theoretic Formulation:} This supply chain can be represented as a two-stage game, in which the supplier acts in the first and the retailers in the second stage. A strategy for the supplier is a price $r\ge 0$ and a strategy for retailer $i$ is a function $q_i:\mathbb R_+\to \mathbb R_+$, which specifies the quantity that retailer $i$ will order for any possible cost $r$. Payoffs are determined via the strategy profile $\(r,\q\(r\)\)$, where $\q\(r\)=\(q_i\(r\)\)_{i=1}^n$. Given cost $r$, the profit function $\po_i\(\q\(r\)\mid r\)$ or simply $\po_i\(\q\mid r\)$, of retailer $i\in N$, is $\po_i\(\q\mid r\)= q_i\(\alpha-q\)^+-rq_i$. For a given value of $\alpha$, the supplier's profit function, $\po_s$ is $\po_s\(r\mid \alpha\)=rq\(r\)$ for $0\le r<\alpha$, where $q\(r\)$ depends on $\alpha$ via $\po_i\(\q\mid r\)$. \vspace*{-0.2cm}

\subsubsection{Continuous Beliefs:} To model the supplier's uncertainty about retail demand, we assume that after his pricing decision, but prior to the order-decisions of the retailers, a value for $\alpha$ is realized from a continuous distribution $F$, with finite mean $\ex\alpha <+\infty$ and nonnegative values, i.e. $F\(0\)=0$. Equivalently, $F$ can be thought of as the supplier's belief about the demand parameter and, hence, about the retailers' willingness-to-pay his price. We will use the notation $\F:=1-F$ for the survival function and $\alpha_L:=\sup{\{r\ge0: F\(r\)=0\}}\ge 0$, $\alpha_H:=\inf{\{r\ge0: F\(r\)=1\}}\le +\infty$ for the support of $F$ respectively. The instance $\alpha_L=\alpha_H$ corresponds to the reference case of deterministic demand. In any other case, i.e., for $\alpha_L<\alpha_H$, the supplier's payoff function $\po_s$ becomes stochastic: $\po_s\(r\)=\ex \po_s\(r\mid \alpha\)$. All of the above are assumed to be common knowledge among the participants in the market (the supplier and the retailers). 

\section{Existing Results}\label{existing}
We consider only subgame perfect equilibria, i.e. strategy profiles $\(r,\q\(r\)\)$ such that $\q\(r\)$ is an equilibrium in the second stage and $q_i\(r\)$ is a best response against any $r$. The equilibrium behavior of this market has been analyzed in \cite{Le17}. In the reference case of deterministic demand, i.e., for $\alpha_L=\alpha_H$, each retailer $i=1,2,\dots, n$ orders quantity $q_i^*\(r\mid \alpha\)=\frac1{n+1}\(\alpha-r\)^+$. Hence, the supplier's payoff on the equilibrium path becomes $\po_s\(r\mid \alpha\)=rq^*\(r\mid \alpha\)=\frac{n}{n+1}r\(\alpha-r\)^+$, for $0\le r$. Maximization of $\po_s$ with respect to $r$ yields that the complete information two-stage game has a unique subgame perfect Nash equilibrium, under which the supplier sells with optimal price $r^*\(\alpha\)=\frac12\alpha$ and each of the retailers orders quantity $q_i^*\(r\)= \frac{1}{n+1}\(\alpha-r\)^+$, $i=1,2,\dots,n$. To proceed with the equilibrium representation in the stochastic case, we first introduce some notation. \vspace*{-0.2cm}

\subsubsection{Generalized mean residual life:}Let $\alpha\sim F$ be a nonnegative random variable with finite expectation $\ex \alpha <+\infty$. The \emph{mean residual life (MRL)} function $\m\(r\)$ of $\alpha$ is defined as
\[\m\(r\):=\ex\(\alpha-r \mid \alpha >r\)=\,\dfrac{1}{\F\(r\)}\int_{r}^{\infty}\F\(u\)\du, \quad\mbox{for } r< \alpha_H \] and $\m\(r\):=0$, otherwise, see, e.g., \cite{Be16}. In \cite{Le17}, we introduced the \emph{generalized mean residual life (GMRL)} function $\e\(r\)$, defined as $\e\(r\):=\frac{\m\(r\)}r$, for $0<r<\alpha_H$, in analogy to the \emph{generalized failure rate (GFR)} function $\mathrm{g}\(r\):=r\mathrm{h}\(r\)$, where $\mathrm{h}\(r\):=f\(r\)/\F\(r\)$ denotes the hazard rate of $F$ and the \emph{increasing generalized failure rate (IGFR)} unimodality condition, defined in \cite{La01} and studied in \cite{La06},\cite{Ba13}. If $\e\(r\)$ is \emph{decreasing}, then $F$ has the \emph{(DGMRL) property}. The relationship between the (IGFR) and (DGMRL) classes of random variables is studied in \cite{Le17}. \vspace*{-0.37cm}

\subsubsection{Market equilibrium:} Using this terminology, we can express the supplier's optimal pricing strategy in terms of the MRL function and formulate sufficient conditions on the demand distribution, under which a subgame perfect equilibrium exists and is unique. 
\begin{theorem}[\cite{Le17}]\label{mainresult}
Assume that the supplier's belief about the unknown, nonnegative demand parameter, $\alpha$, is represented by a continuous distribution $F$, with support inbetween $\alpha_L$ and $\alpha_H$ with $0\le \alpha_L<\alpha_H\le\infty$. 
\begin{itemize}
\item[(a)] If an optimal price $r^*$ for the supplier exists, then $r^*$ satisfies the fixed point equation  
\begin{equation}\label{fixed}r^*=\m\(r^*\)\end{equation} 
\item[(b)] If $F$ is strictly DGMRL and $\ex \alpha^2$ is finite, then in equilibrium, the optimal price $r^*$ of the supplier exists and is the unique solution of \eqref{fixed}. 
\end{itemize}
\end{theorem}

\section{Supply Chain Efficiency}\label{efficiency}
To study the degree in which demand uncertainty affects the realized market profits, we fix a realized demand level $\alpha$ and compare the individual realized profits of the supplier and each retailer between the scenario in which the supplier prices before demand realization and the scenario in which the supplier prices after demand realization. For clarity, the results are summarized in \Cref{fund}.
\begin{center}
\ra{1}
\begin{tabular}{@{}lllll@{}}\toprule
&$\phantom{b}$& \multicolumn{3}{c}{Upstream Demand for Supplier} \\
\cmidrule{3-5}
&&\multicolumn{1}{l}{Uncertain $\alpha\sim F$} &$\phantom{abds}$& \multicolumn{1}{l}{Deterministic $\alpha$} \\
\midrule
\breakcell[l]{Equilibrium \\Wholesale Price} && \breakcell[l]{\\[-0.2cm]\multicolumn{1}{c}{$r^*=\m_F\(r^*\)$}} && \breakcell[l]{\\[-0.2cm]\multicolumn{1}{c}{$r^*=\alpha/2$}}\\
\midrule
&& \multicolumn{3}{c}{Realized Profits in Equilibrium}\\ \midrule
Supplier && $\Pi_s^U=\frac{n}{n+1}r^*\(\alpha-r^*\)^+$ && $\Pi_s^D=\frac{n}{n+1}\(\alpha/2\)^2$\\
Retailer $i$ && $\Pi_i^U=\frac{1}{\(n+1\)^2}\(\(\alpha-r^*\)^+\)^2$ && $\Pi_i^D=\frac{n}{\(n+1\)^2}\(\alpha/2\)^2$\\[0.3cm]
Aggregate && $\Pi_\A^U=\Pi_s^U+\sum_{i=1}^n \Pi_i^U$ && $\Pi_\A^D=\Pi_s^U+\sum_{i=1}^n \Pi_i^U$\\ 
\bottomrule
\end{tabular}
\captionof{table}{Wholesale price in equilibrium and realized profits when the supplier prices under demand uncertainty (left column) and under deterministic demand (right column).}
\label{fund}
\end{center}

\subsection{Price of Uncertainty}\label{PoU}
By \Cref{fund}, for each retailer, we have that $\frac{1}{\(n+1\)^2}\(\(\alpha-r^*\)^+\)^2\ge \frac{1}{\(n+1\)^2}\(\frac\alpha2\)^2$ for all values of $\alpha\ge2r^*$. This implies that for larger values of the realized demand, the retailers are better off if the supplier prices under demand uncertainty. In contrast, the supplier is never better off when he prices under demand uncertainty, as is intuitively expected. Indeed $\frac{n}{n+1}r^*\(\alpha-r^*\)^+\le \frac{n}{n+1}\(\alpha/2\)^2$ for all values of $\alpha$,
\begin{wrapfigure}[14]{r}{0.49\textwidth}
\includegraphics[width=\linewidth]{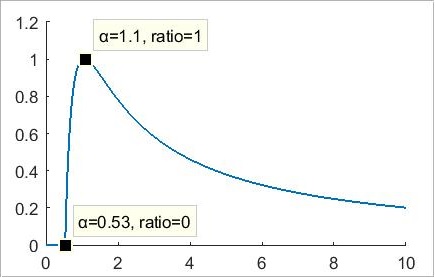}
\caption{\label{supratio}Ratio of the supplier's realized profits with and without demand uncertainty for $\alpha\sim\text{Weibull}\(1,2\)$.}
\end{wrapfigure}
with equality if and only if $\alpha=2r^*$. The ratio of the supplier's realized profit in the scenario with demand uncertainty to the scenario without demand uncertainty is equal to $4\cdot\frac{r^*}{\alpha}\(1-\frac{r^*}{\alpha}\)$ 
and hence it has the shape shown in \Cref{supratio}, independently of the underlying demand distribution.\par
Similar findings are obtained when we compare the market's \textit{aggregate} realized pro\-fits (supplier and retailers) between these two scenarios. This is accomplished via the ratio of aggregate realized market profits under stochastic demand to the aggregate realized market profits under deterministic demand, which we term the realized \emph{Price of Uncertainty} (PoU), motivated by a similar notion that is studied in \cite{Bal13}. Specifically, 
\[\text{PoU}:=\sup_{F\in \mathcal G}\sup_{\alpha}{\left\{\frac{\Pi_{\A}^U}{\Pi_{\A}^D}\right\}}=\sup_{F\in \mathcal G}\sup_{\alpha}{\left\{\frac{\Pi_{s}^U+\sum_{i=1}^n\Pi_i^U}{\Pi_s^D+\sum_{i=1}^n\Pi_i^D}\right\}}\]
in which we restrict attention to the class $\mathcal G$ of nonnegative DGMRL random variables to retain equilibrium uniqueness. Intuitively, one expects the system to perform worse under demand uncertainty which translates to PoU being bounded above by $1$. However, this is not the case as the next Theorem states.
\begin{theorem}\label{thmpou}The realized PoU of the stochastic market is given by $\text{PoU}=1+\mathcal O\(n^{-2}\)$, independently of the underlying demand distribution. The upper bound is attained for realized demand $\alpha^*=\frac{n}{n-1}\cdot2r^*$.
\end{theorem}
\begin{proof}
By \Cref{fund}, a direct substitution yields that the inner ratio is equal to 
\[\frac{\Pi_s^U+\sum_{i=1}^n\Pi_i^U}{\Pi_s^D+\sum_{i=1}^n\Pi_i^D}=\frac{\frac{n}{n+1}r^*\(\alpha-r^*\)^++n\(\frac{1}{n+1}\(\alpha-r^*\)^+\)^2}{\frac{n}{n+1}\(\frac\alpha2\)^2+\frac{n}{\(n+1\)^2}\(\frac\alpha2\)^2}\]
Hence, $\text{PoU}=\sup_{F\in \mathcal G}\sup_{\alpha}{\left\{\frac{4}{\(n+2\)\alpha^2}\(\alpha-r^*\)^+\(\alpha+nr^*\)\right\}}$. For realized demand $\alpha<r^*$, there is a stockout and the market operates worst under demand uncertainty. However, for realized demand values $\alpha>r^*$, the aggregate market profits of the supplier and the retailers may be larger if the supplier prices under demand uncertainty. To see this, we take the partial derivative of the previous ratio with respect to $\alpha$
\[\frac{\partial}{\partial\alpha}\(\frac{4}{\(n+2\)\alpha^2}\(\alpha-r^*\)^+\(\alpha+nr^*\)\)=\frac{4r^*}{\(n+2\)\alpha^3}\(2nr^*-\alpha\(n-1\)\)\]
which shows that the ratio is increasing on $[r^*,\frac{2n}{n-1}r^*)$, and decreasing thereafter. The ratio is maximized for $\alpha=\frac{2n}{n-1}r^*$, yielding a value of $1+\frac{1}{n^2+2n}$, which does not depend on the underlying distribution $F$ and which is larger than $1$ for any number $n$ of competing second-stage retailers. \qed
\end{proof}
The values for which the ratio exeeds $1$, depend on $n$. Specifically, for $n\ge 3$, we have that $\frac{4}{\(n+2\)\alpha^2}\(\alpha-r^*\)^+\(\alpha+nr^*\)\ge1$ for values of $\alpha$ in $[2r^*, \frac{2n}{n-2}r^*]$. In this case, the upper bound decreases to $2r^*$ as $n\to \infty$. For $n=2$, the upper bound is equal to infinity, i.e., the range of $\alpha$ for which the ratio exceeds $1$ is equal to $[2r^*,+\infty)$. In all cases, the lower bound is independent of $n$. Finally, by taking the partial derivative with respect to $n$, we find that the PoU is nondecreasing in $n$ for realized values of $\alpha$ in $[r^*,2r^*]$ and decreasing in $n$ thereafter, again independently of the underlying demand distribution. This is illustrated in \Cref{ncompare}. 
\begin{figure}[!htp]
\includegraphics[width=\linewidth]{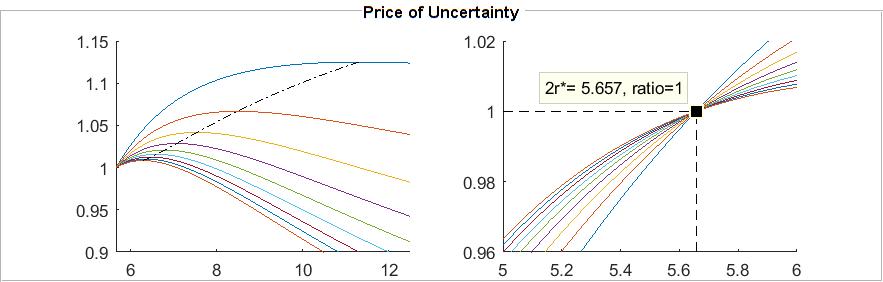}
\caption{Ratio of the aggregate market realized profits with and without demand uncertainty for $n=2$ to $n=10$ with $\alpha\sim \text{Gamma}\(2,2\)$. The dashed curve in the left panel passes through the points $\alpha^*=\frac{n}{n-1}\cdot2r^*$ on which the PoU is attained for each $n$. The curves are decreasing in $n$, i.e. the highest curve corresponds to $n=2$ and the lowest to $n=10$. The right panel shows the behavior of the curves in a neighborhood of their intersection point, $2r^*\approx 5.657$. Prior to the intersection, the ratio is increasing in $n$, whereas after the intersection the ratio is decreasing in $n$. }
\label{ncompare}
\end{figure}\vspace*{-0.7cm}
\subsection{Price of Anarchy}\label{secpoa} As a benchmark, we will first determine the equilibrium behavior and performance of an integrated supply chain. The integrated firms' decision variable is now the retail price $r$, and hence its expected profit $\po_\I$ is given by $\po_\I\(r\)=r\ex\(\alpha-r\)^+=r\m\(r\)\F\(r\)$. By the same argument as in the proof of \Cref{mainresult}, $\po_\I$ is maximized at $r^*=\m\(r^*\)$. In particular, the equilibrium price of both the integrated and non-integrated supplier is the same. Hence, the integrated firm's realized profit in equilibrium is equal to $\Pi_\I^U\(r^*\mid \alpha\)=r^*\(\alpha-r^*\)^+$. \par
In a similar fashio to \cite{Pe07}, we define the realized \emph{Price of Anarchy (PoA)} of the system as the worst-case ratio of the realized profit of the centralized supply chain, $\Pi^U_\I$, to the realized aggregate profit of the decentralized supply chain,
\begin{center}
\ra{1}
\begin{tabular}{@{}lllll@{}}\toprule
&$\phantom{b}$& \multicolumn{3}{c}{Realized Profits in Equilibrium} \\
\cmidrule{3-5}
&&\multicolumn{1}{l}{Uncertain Demand $\alpha\sim F$} &$\phantom{abds}$& \multicolumn{1}{l}{Deterministic Demand $\alpha$} \\
\midrule
Integrated Firm && $\Pi^U_\I=r^*\(\alpha-r^*\)^+$ && $\Pi^D_\I=\(\alpha/2\)^2$\\
\bottomrule
\end{tabular}
\captionof{table}{Realized profits for the integrated firm under the two scenarions. The equilibrium wholesale prices remain the same as in the decentralized market, cf. \Cref{fund}.}
\label{integrated}
\end{center}
$\Pi^U_\D:=\Pi^U_s+\sum_{i=1}^n \Pi^U_i$. Again, to retain equilibrium uniqueness, we restrict attention to the class $\mathcal G$ of nonnegative DGMRL random variables. 
If the realized demand $\alpha$ is less than $r^*$, then both the centralized and decentralized chains make $0$ profits. Hence, we define the PoA as: $\text{PoA}:=\sup_{F\in \mathcal G}\sup_{\alpha>r^*} {\left\{\frac{\Pi^U_\I}{\Pi^U_\D}\right\}}$. We then have
\begin{theorem}\label{poathm} The realized PoA of the stochastic market is given by $\text{PoA}=1+1/n$ independently of the underlying demand distribution. The upper bound is asymptotically attained for $\alpha \searrow r^*$.
\end{theorem}
\begin{proof} By a direct substitution in the definition of PoA, the inner term equals $\frac{\(n+1\)^2}{n}\cdot\(n+\frac{\alpha}{r^*}\)^{-1}$. Since $\(n+\frac{\alpha}{r^*}\)^{-1}$ decreases in the ratio $\alpha/r^*$, the inner $\sup$ is attained asymptotically for $\alpha \searrow r^*$. Hence, 
\begin{equation}\label{poacalc}\text{PoA}=\sup_{F\in \mathcal G}\left\{\frac{\(n+1\)^2}{n}\cdot \(n+1\)^{-1}\right\}=1+\frac1n\end{equation}\end{proof}

\begin{wrapfigure}[18]{r}{0.49\textwidth}
\vspace*{-0.8cm}\includegraphics[width=\linewidth]{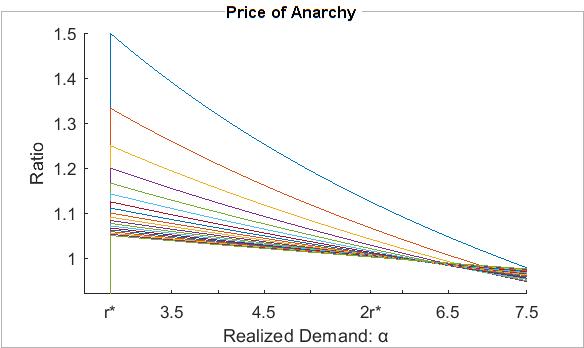}
\caption{\label{poan}Ratio of the integrated firm's to the decentralized market's aggregate profits for $n=2,\dots, 20$ with $\alpha\sim \text{Gamma}\(2,2\)$. For each $n$ the realized PoA is attained as $\alpha \searrow r^*\approx 2.83$. For $n\(\alpha-2r^*\)\le\alpha$, the curves are nonincreasing in $n$, which results in the nonlinearity (with respect to $n$) for values of $\alpha$ in $[2r^*, 3r^*]$.}
\end{wrapfigure}

\Cref{poathm} implies that the market becomes less efficient in the worst-case scenario, i.e., for a realized demand $\alpha \searrow r^*$, as the number of downstream retailers increases. In general, as can be directly inferred by partial differentiation with respect to $n$, for realized values of $\alpha<2r^*$, the inner term of the $\sup$ expression in \eqref{poacalc} is decreasing in $n$. For realized values of $\alpha\ge 2r^*$, the ratio is increasing in $n$ when $n\ge\alpha/\(\alpha-2r^*\)$ and decreasing in $n$ otherwise. These findings are shown graphically in \Cref{poan}.\par
Finally, a similar calcuation yields that the PoA of the deterministic market is equal to $1+\mathcal O\(n^{-2}\)$. The realized demand $\alpha$ simplifies in the inner ratio and hence this upper bound is constant and independent of the demand level. Notably, the PoA in the deterministic market is equal to PoU in the stochastic market, cf. \Cref{thmpou}.

\section{Conclusions}\label{conclusions}
The present study complements the findings of \cite{Le17}, by focusing to the effects of demand uncertainty on market efficiency.\footnote{Along with \cite{Kok18}, the current paper presents preliminary works that appear in full length in \cite{Le18,Leo20,Leo21,Leon21}.} Based on the realized market profits, we measured the effects of uncertainty via the reali\-zed Price of Uncertainty. Counterintuitively, there exist demand levels for which the retailers' and the market's aggregate profits are higher when the supplier prices under demand uncertainty. This is achieved in expense of the supplier's welfare who is always better off under deterministic demand. The realized Price of Anarchy revealed that for any demand level, the integrated chain performs better -- in terms of efficiency -- as the number of competing retailers increases. Upper bounds of inefficiency are attained for lower values of realized demand. Despite these intuitions, the present analysis is limited in extent. Price differentiation and mechanisms that will incentivize retailers to honestly reveal their private information about demand, constitute promising lines of ongoing research.

\bibliographystyle{splncs}

\end{document}